\documentclass[11pt, reqno]{amsart}

\usepackage{amsmath, amssymb, amsthm, cases,  mathtools, multicol, hyperref, subcaption}
\usepackage{amsaddr}
\usepackage[usenames, dvipsnames]{xcolor}

\hypersetup{
	colorlinks=true,
	linkcolor=Cerulean,
	citecolor=Thistle
}

\usepackage{graphicx}
\usepackage{marginnote}
\usepackage{enumerate}
\usepackage{caption}
\usepackage{dsfont}
\usepackage{tikz}

\usepackage{fancyhdr,datetime} \usdate
\date{\today}
\makeatletter
\@namedef{subjclassname@2020}{\textup{2020} Mathematics Subject Classification}
\makeatother

\DeclareMathOperator{\spt}{\mathbf{spt}}

\newtheorem{theorem}{Theorem}[section]
\newtheorem{prop}[theorem]{Proposition}
\newtheorem{corollary}[theorem]{Corollary}
\newtheorem{lemma}[theorem]{Lemma}
\theoremstyle{definition}
\newtheorem{definition}[theorem]{Definition}
\newtheorem{exmp}[theorem]{Example}
\newtheorem{remark}[theorem]{Remark}

\newcommand{\pmat}[1]{\begin{pmatrix}#1\end{pmatrix}}

\newcommand{\dg}[1]{\backslash #1\backslash}

\newcommand{\pt}{\ensuremath{\Tilde{p}}}
\newcommand{\qt}{\ensuremath{\Tilde{q}}}
\newcommand{\R}{\mathbb{R}}

\newcommand{\Nset}{\underline{N}}

\DeclareMathOperator{\diag}{\rm diag}

\newcommand{\RNgt}{{\R^N_+}} 
\newcommand{\RNge}{{\overline{\RNgt}}} 
\newcommand{\calA}{{\mathcal A}}
\newcommand{\calC}{{\mathcal C}}
\newcommand{\calI}{{\mathcal I}}
\newcommand{\calJ}{{\mathcal J}}
\newcommand{\calP}{{\mathcal P}}
\newcommand{\calQ}{{\mathcal Q}}
\newcommand{\AII}{A_{\calI\calI}}
\newcommand{\BII}{B_{\calI\calI}}
\newcommand{\BIJ}{B_{\calI\calJ}}
\newcommand{\BJI}{B_{\calJ\calI}}
\newcommand{\BJJ}{B_{\calJ\calJ}}
\newcommand{\BQQ}{B_{\calQ\calQ}}
\newcommand{\inv}{^{-1}}
\newcommand{\eps}{\epsilon}
\newcommand{\one}{\mathds{1}}

\title
[Exclusion and multiplicity in Lotka-Volterra systems]
{Exclusion  and multiplicity for stable communities in Lotka-Volterra systems}
\author[Won Eui Hong and Robert L Pego]{Won Eui Hong \and Robert L Pego}
\address{Department of Mathematical Sciences,
Carnegie Mellon University, Pittsburgh, PA 15213, USA}
\email{woneuih@andrew.cmu.edu}
\email{rpego@cmu.edu}

\begin{document}

\begin{abstract}
For classic Lotka-Volterra systems governing many interacting species,
we establish an exclusion principle that rules out the
existence of linearly asymptotically stable steady states in subcommunities 
of communities that admit a stable state which is internally $D$-stable.
This type of stability is known to be ensured, e.g., 
by diagonal dominance or Volterra-Lyapunov stability conditions.
By consequence, the number of stable steady states of this type is bounded
by Sperner's lemma on anti-chains in a poset.
The number of stable steady states can nevertheless be very large 
if there are many groups of species that strongly inhibit outsiders
but have weak interactions among themselves.

By examples we also show that in general it is possible for a stable community to 
contain a stable subcommunity consisting of a single species.  Thus a recent empirical
finding to the contrary, in a study of random competitive systems by Lischke and L\"offler
(Theo.~Pop.~Biol.~115 (2017) 24--34), does not hold without qualification.
\vspace{-0.7cm}
\end{abstract}

\keywords{Gause's law, multiple stable equilibria, clique, evolutionarily stable states, biodiversity, $P_0$ matrices}

\subjclass[2020]{92D25, 92D40, 34D20, 91A22}

\maketitle

\tableofcontents
\addtocontents{toc}{\protect\setcounter{tocdepth}{1}}

\section{Introduction}

Lotka-Volterra systems comprise a family of classic and prototypical models in population ecology.
They incorporate nonlinear feedback and regulation mechanisms of 
clear biological importance in a structurally simple way 
that renders them fairly amenable to mathematical analysis.
Partly for this reason they retain value and interest alongside 
models of greater complexity and realism \cite{chesson2000mechanisms}.

Based on such a model, Volterra in 1928 \cite{volterra1928variations} demonstrated 
that two species exploiting a common resource cannot stably coexist.  
Volterra's findings strongly influenced the development of competitive exclusion principles 
and ecological niche theory by Gause \cite{gause1934experimental}, 
Hutchinson \cite{hutchinson1957} and many others.
The notion of competitive exclusion in general Lotka-Volterra competition models 
was later investigated mathematically rather thoroughly and was found to be subject to a number of
limitations \cite{armstrong1980competitive,mcgehee1977}.  
Moreover it was discovered that, in principle, dynamics in such models can be arbitrarily
complicated, admitting time-periodic and even chaotic behavior in systems with
only a few species \cite{may1975,gilpin1975limit,smale1976}.  Nevertheless, the
concept of competitive exclusion remains valuable and influential in ecology,
as recently noted by Pocheville \cite{pocheville2015}.

The present work is motivated by investigations regarding the number of
alternative stable steady states that a given (or typical) Lotka-Volterra
competition model may admit.  
Such investigations relate to a variety of significant issues in ecology, 
such as whether a given local community of species might be 
 susceptible to invasion by a species that is not yet present,
how a particular assembly of species may have come to co-exist,
or whether different outcomes may have been possible 
based on different histories of invasion.  
See
\cite{gilpin1976multiple,
may1977,
case1979global,
ings2009ecological,
svirezhev2008nonlinearities,
lawmorton1993alternative,
kokkoris1999patterns}
for a small selection of papers that address such issues. 

Recently, Lischke and L\"offler \cite{lischke2017finding} developed
numerical methods for efficiently finding all the possible stable steady states
in a given Lotka-Volterra model. They carried out extensive numerical experiments to
analyze a class of random competitive systems for up to 60 species, examining the
effect of relative sizes of competition coefficients on the number and type of stable equilibria.
In a small percentage of cases they find more than 30 alternative stable steady states.
In addition they mention an empirical finding related to an exclusion principle.
Loosely paraphrasing, they found that no species which 
forms by itself a single-species stable community 
was ever observed to be a member of any alternative stable community.
If this were always true, then one could often greatly simplify the search for stable
communities by studying the stability of the simple single-species steady states.

{\it Exclusion.}
Below, we establish several community exclusion principles 
related to these findings.
We prove that a generalization of the empirical Lischke-L\"offler
exclusion principle is valid in certain circumstances.  
In the special case of symmetric (or diagonally symmetrizable) interspecific interaction coefficients,
it is universally the case for all Lotka-Volterra systems, whether competitive or not,
that any stable community can neither contain nor be contained in any other such.
(See Corollary~\ref{cor:LL}.)

In the general case without symmetry, we show (Theorem~\ref{t:1exclusion}) that 
no two stable communities can differ by exactly one species. 
Furthermore, any community that is ``internally $D$-stable'' 
does not admit any stable subcommunity (Theorem~\ref{t:Dexclusion}).
The general concept of $D$-stability has been
much studied for constant-coefficient linear systems of differential equations in general and 
linearized Lotka-Volterra systems in particular;
see \cite{cross1978,logofet2005,kushel2019} and \cite[Sec.~15.6]{HofbauerSigmund}.
Practical criteria that precisely characterize $D$-stability are not known 
in general, but sufficient criteria include stability due to diagonal dominance, 
and Volterra-Lyapunov stability, meaning sign-definiteness of an associated quadratic form after diagonal scaling~\cite[Prop.~1]{cross1978}.

A still-open conjecture of Hofbauer and Sigmund states that an equilibrium state
involving all species of a Lotka-Volterra system is globally attracting 
if the interspecific interaction matrix is $D$-stable.
Theorem~\ref{t:Dexclusion} supports this conjecture insofar as it implies that no other
equilibrium involving fewer species can be locally attracting. 

The empirical Lischke-L\"offler exclusion principle for Lotka-Volterra competitive systems
turns out not to be valid without some qualification, however.  By example, we show in Section~\ref{ss:counterex} 
that a single species forming a stable equilibrium by itself can be contained in a larger stable community. 
It is plausible that such systems may be rare in typical random ensembles.  
If that is the case, an exclusion property for stable subcommunities may be expected, though not guaranteed.

{\it Multiplicity.}  The maximum number of stable steady states 
that can co-exist in Lotka-Volterra systems is an interesting quantity to consider,
and can be limited by community exclusion principles such as we study here.
If all interspecific interactions are symmetric, or all stable states are internally $D$-stable, then
the maximum number of stable equilibria is bounded via Sperner's lemma for 
anti-chains in posets \cite{lubell1966}; see Section~\ref{ss:sperner} below.
For $N$ species with $N$ large, this bound is approximately $2^N\sqrt{2/\pi N}$,
which is a number somewhat smaller than $2^N$, the number of all subsets of the $N$ species,
but one that still grows exponentially fast in $N$.
We do not know whether the bound from Sperner's lemma is sharp.

It is true that exponentially many alternative stable subcommunities are
possible in principle, however.  
Particular highly symmetric examples can be constructed 
similar to how cliques in graphs have been used
to form stable states in game theory~\cite{vickers1988patterns}
and continuous-time models of allele selection in population genetics 
(replicator equations with symmetric payoff matrix) \cite[p.~255]{HofbauerSigmund}.

In Section~\ref{s:multiplicity} we describe and generalize this construction and establish
quantitative criteria capable of ensuring that large numbers of 
alternative stable subcommunities are possible in certain Lotka-Volterra systems 
for $N$ species.
This can happen when many different communities can be formed consisting of species
that compete weakly with each other while strongly inhibiting outsiders.
Our criteria may have relevance for some biological systems.
E.g., certain recent works \cite{coyte2015ecology,goldford2018emergent} suggest that
there may be common patterns of interaction among the many alternative species in
naturally occurring microbiomes. 
In particular, weak interactions may be predominant in the 
microbiome of the human gut---a community comprising hundreds of species of bacteria---but 
the presence of some strongly competitive interactions can have a stabilizing 
effect~\cite{coyte2015ecology}. 

{\it Relation to evolutionary game theory.}
It is well known that there is an equivalence between the dynamics of a given Lotka-Volterra system
and those of a corresponding family of \emph{replicator equations} in evolutionary game theory. 
A rather extensive body of work exists concerning 
exclusion principles and multiplicity for stable states in replicator equations. 
Some of the findings in this opus carry back readily to Lotka-Volterra systems.
For others, their game-theoretic meaning has no evident significance in the Lotka-Volterra context.
The degenerate nature of the correspondence can also get in the way.

We will make a detailed comparison of our findings with corresponding results 
on replicator equations in Section~\ref{s:ESS}.
Of special significance is 
the game-theoretic notion of an \emph{evolutionarily stable state (ESS)},
which has been extensively explored following its introduction by 
Maynard Smith and Price \cite{smith1973logic} in an analysis of animal conflict.
Each ESS is a locally attracting steady state for replicator dynamics,
but the reverse is not generally true for non-symmetric payoff matrices.
The supports of ESSs are known to satisfy the same type of exclusion principle 
(a non-containment property known as the Bishop-Cannings theorem~\cite{bishop1976models}) 
as we establish here for internally $D$-stable equilibria in Lotka-Volterra systems. 

We show that the ESS notion does not correspond to internal $D$-stability under the 
replicator--Lotka-Volterra equivalence, however. 
Nor are ESSs invariant under diagonal scalings natural to Lotka-Volterra systems.
An interesting and extensive understanding of the multiplicity and
patterns of possible ESSs for large numbers of strategies has been achieved;
the recent paper~\cite{bomze2020ess} has pointers to much relevant literature.
Yet it remains unclear whether corresponding results hold which are meaningful for Lotka-Volterra systems.

\section{Lotka-Volterra systems and notions of stability}

\subsection{Governing equations} Lotka-Volterra systems model the time evolution of 
the populations $p_i$ of a finite set of $N$ species indexed by $i\in\Nset:=\{1,2,\ldots,N\}$. 
With $'$ denoting the time derivative, the governing differential equations take the form
    \begin{equation}\label{eq:generalLV}
        p'_{i} = p_{i}(a_i - \sum_{j\in\Nset} B_{ij}p_j) \,, \quad i\in\Nset.
    \end{equation}
Here $a_i$ represents an intrinsic growth rate for species $i$ in the limit when all populations are small,
and $B_{ij}$ is a coefficient which, if positive, induces a competitive or inhibiting effect 
of the presence of species $j$ on the growth of species $i$.
Throughout this paper we take the coefficients $a_i$ and $B_{ij}$ to be constant in time.

Almost exclusively, our interest is in solutions of \eqref{eq:generalLV}
belonging to the state space 
$\RNge= \{p\in \mathbb{R}^N : p_i\ge 0\;\forall i\}$,
since negative species populations are normally not meaningful.
It is convenient that this space is invariant for solutions of system
\eqref{eq:generalLV}.

Given a state $p\in\RNge$, the \emph{community supporting} $p$ will
refer to the set of species $i$ for which $p_i>0$.
Mathematically this is the support, denoted $\spt p=\{i\in\Nset:p_i>0\}$.
The community supporting a solution $t\mapsto p(t)$ is time-invariant,
since $p_i(t)$ is either always positive or always zero.

In order to write this system in a convenient matrix-vector
form, we define
\[
\dg{p} = \diag(p_1,\ldots,p_N)
\]
to denote the diagonal matrix with successive diagonal entries 
$p_1,\ldots,p_N$.
With this notation, equation \eqref{eq:generalLV} takes the form
    \begin{equation} \label{e:LV2}
       p' = \dg{p}(a - Bp)   \,.
    \end{equation}

\subsection{Equilibria, linearization, scaling}

A vector $\pt\in\RNge$ is a steady state (or equilibrium) for the system~\eqref{eq:generalLV}
if and only if  
\begin{equation}\label{e:eq1}
a_i-(B\pt)_i=0 \quad\text{for each $i\in\spt\pt$.}
\end{equation}
We will analyze the system in block form with respect to the support community $\calI=\spt\pt$ and 
its complement $\calJ = \Nset\setminus\calI$, via the notation
\[
p = \begin{pmatrix}p_\calI\\ p_\calJ\end{pmatrix},\qquad
a = \begin{pmatrix}a_\calI\\ a_\calJ\end{pmatrix},\qquad
B = \begin{pmatrix}\BII &\BIJ\\ \BJI &\BJJ\end{pmatrix}.
\]
Then $\pt_\calJ=0$, and  \eqref{e:eq1} means that $a_\calI = \BII\pt_\calI$.
Thus for any community $\calI\subseteq\Nset$, if $\BII$ is invertible then 
the community $\calI$ supports at most one steady state in $\RNge$.

The linearized equation of evolution for small perturbations $q$ around the steady state $\pt$ 
takes the form 
\begin{align} \label{e:linear}
             q'&=Aq,
\end{align}
where the constant matrix $A$ is explicitly given by 
\[
A_{ij} = \begin{cases} 
-\pt_i B_{ij} & \text{for $i\in\calI$ and any $j\in\Nset$},\\
a_i-(B\pt)_i & \text{for $i\notin\calI$ and $j=i$},\\
0 &\text{for $i\notin\calI$ and $j\ne i$.}
\end{cases}
\]
In block form using the diagonal-matrix notation $\dg{p}$ above, we can write
\begin{equation}\label{d:blockL}
 A = \dg{a-B\pt}-\dg{\pt} B = \begin{pmatrix} 
  -\dg{\pt_\calI} \BII & -\dg{\pt_\calI} \BIJ \\
   0 & \dg{a_\calJ-\BJI\pt_\calI}
\end{pmatrix} .
\end{equation}

Diagonal scaling will sometimes be used for our analysis. 
If $D=(d_{ij})$ is a diagonal matrix with positive diagonal entries $d_{ii}>0$,
the change of variables $p=D\hat p$ maps \eqref{e:LV2} to the system
\begin{equation}\label{e:Dscale}
 \hat p' = \dg{\hat p}(a-\hat B\hat p), \qquad \hat B=B D,
\end{equation}
having scaled columns, with $\hat B_{ij}=B_{ij} d_{jj}$. 
If $d_{ii}=\pt_i$ for $i\in\calI=\spt\pt$, 
the scaled equilibrium is uniform over $\calI$, 
with $\hat p_i=1$ if and only if $i\in\calI$, which we write as
$\hat p = \one_\calI$.

\subsection{Notions of stability}

\subsubsection{Matrix conditions}\label{sss:matrix}
We recall a few standard definitions for matrices that relate to the stability properties of 
the linear system~\eqref{e:linear}~\cite{cross1978,logofet2005,HofbauerSigmund}. 
\begin{definition} Let $A$ be a real $N\times N$ matrix.
        \begin{enumerate}
            \item $A$ is \emph{stable} if every eigenvalue of $A$ has negative real part.
            \item $A$ is \emph{$D$-stable} if $DA$ is stable for all diagonal $D>0$.
            \item $A$ is \emph{Volterra-Lyapunov stable} (VL-stable) if there exists some 
	    diagonal $D>0$ for which $DA+A^TD<0$, or equivalently 
	    $\langle x, DAx\rangle <0$ for all $x\in\R^N\setminus\{0\}$.
        \end{enumerate}
    \end{definition}
Here the notation $S>0$ (resp. $S\ge0$ or $S<0$) for a real symmetric matrix $S$
means $S$ is positive definite (resp. positive semidefinite or negative definite),
and $\langle\cdot,\cdot\rangle$ denotes the standard inner product in $\R^N$.

It is known that Volterra-Lyapunov stability implies $D$-stability;
see \cite[Prop.~1]{cross1978}.
Of course, $D$-stability implies stability.  
In case $A$ is symmetric, the three notions are equivalent, since stability is
equivalent to negative definiteness. 

The three notions are equivalent also in case $A$ is \emph{$D$-symmetrizable}, 
meaning $D_1 A D_2$ is symmetric for some positive diagonal $D_1$, $D_2$. 
For if $A$ is stable and $D=D_2^{-1/2}D_1^{1/2}$, then the symmetric matrix 
$S=DAD\inv<0$, hence $2 DSD= D^2A+A^TD^2<0$, thus $A$ is VL-stable.

\subsubsection{Linear stability}

For Lotka-Volterra systems in general, it is arguably natural to study stability 
restricted to the invariant state space $\RNge$. In linearly degenerate cases this leads to some subtleties. 
E.g., in the simple example of the system
\[
\begin{cases}
    x'=\beta x^2,\\
    y'=-y.
\end{cases}
\]
the origin is clearly not stable in $\R^2$ whenever $\beta\ne0$, but if $\beta<0$ it is asymptotically stable with respect to dynamics restricted to the  quadrant $\overline{\R^2_+}$. 

Our main results concern equilibria $\pt$ which are stable in the nondegenerate
sense of being \emph{linearly asymptotically stable in $\R^N$.}
This means that $q(t)\to0$ as $t\to\infty$ for every solution of \eqref{e:linear} in $\R^N$.  
Arguably, we should compare this to the putatively weaker property that 
$q(t)\to0$ as $t\to\infty$ for just those solutions of \eqref{e:linear} 
for which $\pt+\eps q\in\RNge$ for sufficiently small $\eps>0$. 
We say \emph{$\pt$ is linearly asymptotically stable in $\RNge$} provided
this is the case, which simply means
\begin{equation}\label{c:qi}
q_i\ge0 \quad\mbox{whenever}\quad \pt_i=0.
\end{equation}
These two notions of linear asymptotic stability turn out to be equivalent, however,
so we will not need to refer to the second notion in what follows.

    \begin{prop}\label{p:easd}
       Let $\pt\in\RNge$ be an equilibrium  for the Lotka-Volterra system \eqref{eq:generalLV}. 
        Then the following are equivalent:
        \begin{itemize}
         \item[(i)] $\pt$ is linearly asymptotically stable in $\R^N$.
         \item[(ii)] $\pt$ is linearly asymptotically stable in $\RNge$.
        \item[(iii)] $A$ is stable.
        \end{itemize}
\end{prop}
\begin{proof}
    The equivalence of conditions (i) and (iii) is well known, and (i) implies (ii). 
    If (i) fails to hold, then the matrix $A$ in \eqref{e:linear} has some eigenvalue with non-negative real part. By consequence, each solution satisfying $q(t)\to0$ as $t\to\infty$ must lie in a \emph{strict} subspace of $\R^N$, which cannot contain an open set in $\R^N$. Since \eqref{c:qi} allows an open set of perturbations, we can conclude that (ii) implies (i).
\end{proof}

For brevity, we say $\pt$  is \emph{strictly stable} 
if $\pt$ is linearly asymptotically stable.
We call $\calI$ a \emph{stable community} if it supports a strictly stable equilibrium $\pt$. 
\subsection{Nonlinear stability}
There is a substantial body of literature 
regarding the nonlinear stability of Lotka-Volterra equilibria, especially 
with respect to solutions with positive population $p_i$ for every species considered,
so that $p(t)\in \RNgt= \{p\in \mathbb{R}^N : p_i> 0\;\forall i\}$ for all $t$.
For example, the books of Goh \cite{goh1980management}, Takeuchi~\cite{takeuchi}
and Hofbauer and Sigmund~\cite{HofbauerSigmund} contain much information. 
We will mainly leave aside issues concerning degenerate cases that involve
eigenvalues with zero real part.

As is well known, condition (i) above ensures that the equilibrium $\pt$ 
is locally asymptotically stable, 
i.e., it attracts all solutions of \eqref{eq:generalLV} in a small enough neighborhood in $\R^N$.
Also well known is the fact that 
\emph{$\pt$ globally attracts all solutions in $\RNgt$ if $-B$ is Volterra-Lyapunov stable};
see \cite{goh1977global} and \cite[p.~191]{HofbauerSigmund}.

Hofbauer and Sigmund have conjectured in \cite[p.~200]{HofbauerSigmund} that 
$\pt$ globally attracts all solutions in $\RNgt$
if $A$ is $D$-stable.
To our knowledge, this conjecture remains open.

\subsection{Internal stability}\label{ss:internal}
Given an equilibrium $\pt$ with support community $\calI$, it is often natural
to consider its stability with respect to solutions supported by the same community.
\begin{definition}\label{d:internal}
Let $\pt$ be an equilibrium state for the Lotka-Volterra system \eqref{eq:generalLV},
and let $\calI=\spt\pt$ be its support community.
We say:
\begin{itemize}
\item $\pt$ is \emph{internally stable} if $-\dg{\pt_\calI}\BII$ is stable.
\item $\pt$ is \emph{internally $D$-stable} if $-\BII$ is $D$-stable.
\item $\pt$ is \emph{internally VL-stable} if $-\BII$ is Volterra-Lyapunov stable.
\end{itemize}
\end{definition}
We will also call the \emph{community} $\calI$ internally stable (resp. $D$- or VL-stable)
if it supports some equilibrium $\pt$ 
which is internally stable (resp. $D$- or VL-stable).
Note that if $\pt$ is internally stable, then $\BII$ is nonsingular and thus $\pt$
is the \emph{unique} equilibrium state supported by $\calI$, determined by 
$\pt_\calI=\BII\inv a_\calI$.  

If $\pt$ is internally VL-stable, then it attracts all solutions of \eqref{eq:generalLV}
having the same support community $\calI$.
If $\pt$ is internally ($D$-)stable, it attracts all nearby solutions of \eqref{eq:generalLV}
having the same support community.

These notions of internal stability say nothing about the behavior of solutions
under perturbations which introduce species external to the community 
$\calI$ supporting $\pt$.
Due to the block structure of the linearized system in 
\eqref{e:linear}, this behavior is evidently determined by the 
sign of $(a-B\pt)_i$ for $i\notin\calI$.
It will be convenient to consider this concept for species belonging to 
some given community ${\mathcal Q}\subseteq\Nset$.

\begin{definition}
Let ${\mathcal Q}\subseteq\Nset$, and 
let $\pt$ be an equilibrium for \eqref{eq:generalLV} 
with support community $\calI$ contained in ${\mathcal Q}$.
We say $\pt$ is \emph{$\mathcal{Q}$-stable} if 
$\pt$ is internally stable and 
\begin{equation}\label{d:external}
a_i - (B\pt)_i <0 \quad\text{for all $i\in {\mathcal Q}\setminus\calI$.}
\end{equation}
\end{definition}
Informally, this notion ensures that the (internally stable) community $\calI$ 
that supports $\pt$ is 
stable against (infinitesimal) invasions by other species in ${\mathcal Q}$. 
In particular, if we take ${\mathcal Q}=\Nset$,
it is straightforward to see that we have the following.
\begin{lemma}
Let $\pt$ be an equilibrium state for system \eqref{eq:generalLV}.
Then $\pt$ is strictly stable (i.e., linearly asymptotically stable) if and only if 
it is $\calQ$-stable with $\calQ=\Nset$.
\end{lemma}

\section{Exclusion principles for stable communities}

\subsection{Statements of main results}
Recall that a fundamental result from the book of Hofbauer and Sigmund \cite[Sec.~15.3]{HofbauerSigmund}
states that if the full matrix $-B$ is Volterra-Lyapunov stable, then the Lotka-Volterra system
\eqref{eq:generalLV} admits a unique globally stable equilibrium state in
$\RNge$. (Also see \cite{LiuCaiSu2015} in case $B$ is positive definite.)
With weaker conditions on $B$, it becomes possible that the system admits many more stable equilibria,
and this can have interesting consequences for explaining 
the diversity and historical development of ecological communities
\cite{gilpin1976multiple,case1979global,kokkoris1999patterns,goldford2018emergent}.
Thus it is interesting to identify any restriction on the composition of stable communities,
such as a competitive exclusion principle, which may follow from the nature of interspecific interactions.

For example, one result that follows directly from the global stabilty theorem for 
Volterra-Lyapunov stable matrices $-B$ in \cite[Sec.~15.3]{HofbauerSigmund} is the following:
\begin{theorem}\label{t:VLexclusion}
For any community $\calQ\subseteq\Nset$, if the principal submatrix $-B_{\calQ\calQ}$ is Volterra-Lyapunov stable,
then there is a unique equilibrium $\qt\in\RNge$ with support contained in $\calQ$ that is $\calQ$-stable.
This equilibrium $\qt$ attracts all solutions of \eqref{eq:generalLV} with support community $\calQ$.
\end{theorem}
This follows by simply restricting the equations in \eqref{eq:generalLV} to species $i\in\calQ$
and setting $p_j=0$ for $j\notin\calQ$.
In case the equilibrium $\qt$ is given and $\qt_i>0$ for all $i\in\calQ$,
the global stability follows from an argument going back to Volterra~\cite[\S10.2]{volterra1928variations}
using the strict monotonicity of $F(p(t))$ for the relative entropy functional given by
\begin{equation}\label{def:Fp}
F(p) = \sum_{i\in\calQ} d_i(\qt_i \log\frac{\qt_i}{p_i} + p_i - \qt_i),
\end{equation}
with coefficients $d_i>0$ determined by VL-stability. 
See also \cite{goh1977global,LiuCaiSu2015}.

The empirical finding of Lischke and L\"{o}ffler \cite{lischke2017finding},
if valid, would provide another powerful example of an exclusion principle. 
In their extensive computational experiments,  they found (in the present
terminology) that no stable single-species community was ever a
subcommunity of any other stable community.  
As it is easy to check the stability of equilibria supported by a single species,
Lischke and L\"offler could use this principle to greatly simplify the search for 
all stable communities in large systems.

A quite general exclusion principle for stable communities 
of the Lischke-L\"offler type is in fact valid,
under the condition that the interaction matrix $B$ is $D$-symmetrizable. 
\begin{corollary}\label{cor:LL}
Suppose $B$ is $D$-symmetrizable, and $\calI$ is a community supporting a strictly stable equilibrium 
$\pt$ for \eqref{eq:generalLV}.
Then no other community contained in or containing $\calI$ can support a strictly stable
equilibrium $\qt$.
\end{corollary}
\begin{proof}
    Suppose $\pt$ and $\qt$ are both strictly stable and $\calI\subseteq\calQ=\spt\qt$.
Then each is internally stable, and since $B$ is $D$-symmetrizable, each is internally VL-stable. 
In particular, $-B_{\calQ\calQ}$ is VL-stable, so by the Theorem,
$\qt$ attracts all solutions with support community $\calQ$.
But if $\qt\ne\pt$, this contradicts the strict stability of $\pt$, 
which makes $\pt$ locally asymptotically stable in $\R^N$.
\end{proof}

In the terminology introduced at the end of the last section, Corollary~\ref{cor:LL} states that
if $B$ is $D$-symmetrizable, \emph{different stable communities cannot completely overlap.}
This strong subcommunity exclusion principle does not hold in general
in the absence of symmetrizability or any special stability properties. 
However, we find that it does always hold for communities that differ by only one species.

\begin{theorem}\label{t:1exclusion}
No two stable communities can differ by exactly one species. 
I.e., if $\calI\subset\Nset$ and $x\in\Nset\setminus\calI$, then two equilibrium states
with supporting communities $\mathcal{I}$ and $\mathcal{I}\cup\{x\}$ cannot 
both be strictly stable. 
\end{theorem}

Finally, we are able to exclude complete overlap for stable communities
under a weaker assumption than in Theorem~\ref{t:VLexclusion}.
In particular, the assumption that the larger community is internally $D$-stable suffices.

\begin{theorem}\label{t:Dexclusion}
Suppose $\calQ$ is a community supporting an internally $D$-stable equilibrium $\qt$.
Then no subcommunity of $\calQ$ can support any different equilibrium state which is $\calQ$-stable.
\end{theorem}

The proofs of Theorems~\ref{t:1exclusion} and \ref{t:Dexclusion} will be provided in subsections~\ref{ss:sstable} and \ref{ss:intDstable} below.  
The notion of internal $D$-stabilty seems to arise naturally from \eqref{e:linear}--\eqref{d:blockL}, 
since the stability of the block $-\BII$ is unaffected by any positive diagonal scaling.
Despite a long history of investigation, however, computationally effective
criteria that completely characterize $D$-stability are presently known only for $N\le4$ \cite{cross1978,logofet2005}.
One simple criterion that is sufficient to ensure the matrix $-\BII$ is $D$-stable,
though, which follows from Gershgorin's circle theorem,
is the diagonal dominance condition 
\begin{equation}\label{e:ddominance}
B_{ii}>\sum_{j\in\calI\setminus\{i\}} |B_{ij}| \quad\text{for all $i\in\calI$.}
\end{equation}
This condition ensures $-\BII$ is VL-stable also --- see Remark~\ref{r:d2vl} below, 
and \cite[p.~87]{logofet2005} for a more general result.

Theorem~\ref{t:Dexclusion} excludes the complete overlap of a stable community
by any larger internally $D$-stable community, stable or not.
This would appear to support the conjecture of Hofbauer and Sigmund
\cite[p.~200]{HofbauerSigmund} regarding the global stability of an equilibrium 
with full support $\calI=\Nset$ when $-B$ is $D$-stable.  
For if such an equilibrium is not a global attractor in $\RNgt$, then there cannot be any other
strictly stable equilibrium in the system.  
Our present results leave open the possibility, however, that there could be some 
other equilibrium that is degenerately (semi-)stable,
or there could be an open set in $\RNgt$ with non-convergent dynamics.

In the most general case without symmetry, we find that an exclusion principle 
for stable sub- or super-communities does not always hold.
Here is a basic counterexample.
\begin{exmp}[Failure of subcommunity exclusion]\label{ex:fail ex}
One can check that if 
\begin{equation}\label{d:Ba1}
     B= \begin{pmatrix}
        1 & 1 & 1\\
        2 & 1 & 3\\
        3 & 1 & 4
     \end{pmatrix},
     \quad
a=\begin{pmatrix}
        4 \\ 7 \\9
     \end{pmatrix},
\end{equation}
then the two different equilibrium states of \eqref{eq:generalLV} given by
\[
 \pt=\begin{pmatrix}
        4 \\ 0 \\0
     \end{pmatrix},
     \qquad
     \qt=\begin{pmatrix}
        1 \\ 2 \\1
     \end{pmatrix},
\]
with completely overlapping supports, are both strictly stable.
\end{exmp}
A key property of the matrix $B$ in this example is that $-\dg{\qt}B$ is stable but not $D$-stable.
(In particular it is not a $P_0$ matrix, see subsection~\ref{ss:intDstable} below.)
Here the single-species equilibrium $\pt$ is linearly stable in
a strong sense: the matrix $A$ in \eqref{e:linear}--\eqref{d:blockL} is upper triangular with negative diagonal.
The existence of a stable supercommunity is only possible because $B$ is not $D$-stable.
In subsection~\ref{ss:counterex} below we will examine this more carefully 
and show that such examples can be produced for any number of species $N\ge3$.

\subsection{Exclusion for internally VL-stable states}\label{ss:VLstable}
For the convenience of the reader, we prove Theorem~\ref{t:Dexclusion} first in the special case
when the equilibrium $\qt$ is internally VL-stable, i.e., when the principal submatrix $-B_{\calQ\calQ}$ is VL-stable.
Of course, in this case the more general result of Theorem~\ref{t:VLexclusion} holds, but the following proof,
related to the dissipation rate of the Lyapunov function $F(p)$ in \eqref{def:Fp},
is simple and self-contained. 

\begin{proof}[Proof of Theorem~\ref{t:Dexclusion} for internally VL-stable communities]
Let $\calQ\subseteq\Nset$ be a community supporting an internally VL-stable equilibrium $\qt$,
and suppose $\pt$ is a $\calQ$-stable equilibrium with supporting community $\calP=\spt\pt\subseteq\calQ$.
Note that 
\[
a_i - (B\pt)_i 
  \begin{cases}
  =0 &\text{for all $i\in \calP$,}\\ 
  <0 &\text{for all $i\in\calQ\setminus\calP$,}
  \end{cases}
\]
while $a_i-(B\qt)_i=0$ for all $i\in\calQ$.
    Let $D$ be a positive diagonal matrix making the quadratic form of $DB_\mathcal{QQ}$ positive definite,
    and let $d_i=D_{ii}$. Then
    \[
    K := \sum_{i\in\calQ} \Tilde{q}_id_i(a_i - (B\pt)_i) + \sum_{i\in\calQ}\pt_id_i(a_i -(B\Tilde{q})_i) \leq 0,
    \]
   while on the other hand, since $0=\qt_i(a-B\qt)_i=\pt_i(a-B\pt)_i$ for all $i$,
    \begin{align*}
        K &= \sum_{i\in\calQ} \Tilde{q}_id_i((B\Tilde{q})_i - (B\pt)_i) + \sum_{i\in\calQ}\pt_id_i((B\pt)_i -(B\Tilde{q})_i)\\
            &= (\pt -\Tilde{q})_\mathcal{Q}\cdot DB_{\mathcal{QQ}}(\pt -\Tilde{q})_\mathcal{Q} \geq 0.
    \end{align*}
    Thus $\pt = \Tilde{q}$.
    \end{proof}
   \begin{remark}
The same proof also proves that if $-\BQQ$ is any VL-stable principal submatrix of $B$,
then there is at most one equilibrium with supporting community contained in $\calQ$ that 
satisfies the (degenerate) condition $a_i-(B\qt)_i\le0$  for all $i\in\calQ$.  
This statement follows from stronger results proved in \cite[Sec.~15.4]{HofbauerSigmund}.
   \end{remark} 

\subsection{Exclusion for strictly stable states}\label{ss:sstable}
The proofs of Theorems~\ref{t:1exclusion} and \ref{t:Dexclusion} make use of Schur complements.
If $B$ is a square matrix with block representation
\[
        B= \left(
        \begin{array}{cc}
        B_\mathcal{II}  & B_\mathcal{IJ} \\ 
        B_\mathcal{JI} & B_\mathcal{JJ}
        \end{array}
        \right),
\]
and $B_\mathcal{II}$ is {invertible}, the \emph{Schur complement of $B_\mathcal{II}$ in $B$} is defined by 
\[
B/B_\mathcal{II} := B_\mathcal{JJ}-B_\mathcal{JI}B_\mathcal{II}^{-1}B_\mathcal{IJ}.
\]
Denoting the identity matrix by $I$, block row operations yield
        \[
        \left(
        \begin{array}{cc}
        I  & 0 \\ 
        -B_\mathcal{JI} B_\mathcal{II}^{-1} & I
        \end{array}
        \right)
        \left(
        \begin{array}{cc}
         B_\mathcal{II} & B_\mathcal{IJ} \\ 
        B_\mathcal{JI} & B_\mathcal{JJ}
        \end{array}
        \right)=
        \left(
        \begin{array}{cc}
        B_\mathcal{II} & B_\mathcal{IJ} \\
        0 & B/\BII 
        \end{array}
        \right),
        \]
evidently the \emph{Schur determinant formula} holds:
        \[
        \text{det}B=\text{det}B_\mathcal{II}\;\text{det}(B/B_\mathcal{II}).
        \]

\begin{proof}[Proof of Theorem~\ref{t:1exclusion}]
Let $\pt$ and $\qt$ be strictly stable equilibria for the Lotka-Volterra system~\eqref{eq:generalLV}
with respective support communities $\calI$ and $\calQ=\calI\cup\{x\}$, where $x\notin\calI$.
Note that $a_\calI=\BII\pt_\calI$, and due to the external stability condition \eqref{d:external}, 
\[
0> a_x-(B\pt)_x = a_x - B_{xI}\BII\inv a_\calI.
\]
Since $a_\calQ=\BQQ\qt_\calQ$, this is equal to 
\[
(B_{x\calI}\qt_\calI+B_{xx}\qt_x) - B_{x\calI}\BII\inv(\BII\qt_\calI+B_{\calI x}\qt_x) = (\BQQ/\BII)\qt_x.
\]
Thus $0>(\BQQ/\BII) = (\det\BQQ)/(\det\BII)$. The internal stability of $\pt$ and $\qt$ 
imply that all the eigenvalues of the matrices $\dg{\pt_\calI}\BII$ and $\dg{\qt_\calQ}\BQQ$
have positive real part,
hence both $\det\BQQ$ and $\det\BII$ are positive.
This yields a contradiction.
\end{proof}

\subsection{Proof for internally \texorpdfstring{$D$}{D}-stable states}\label{ss:intDstable}
A key ingredient in our proof of Theorem~\ref{t:Dexclusion} is that $D$-stable matrices enjoy a property which
behaves nicely under restriction to principal submatrices and their Schur complements.
Firstly, it is known \cite[p.~256]{cross1978} that for any $D$-stable matrix $A$, $-A$ has the following $P_0$ property.
    \begin{definition}
    $A\in \R^{N\times N}$ is a $P_0$ matrix if every principal minor of $A$ is nonnegative.
\end{definition}

\begin{lemma}
Schur complements in any nonsingular $P_0$ matrix are also $P_0$. 
\end{lemma}
\begin{proof}
Let $B$ be an $N\times N$ nonsingular $P_0$ matrix with block representation as above. In order to consider principal submatrices of $B/B_\mathcal{II}$, choose $\mathcal{K}\subseteq \mathcal{J}$ and let $\mathcal{I'}=\calI\cup\mathcal{K}$. Then $(B/B_\mathcal{II})_\mathcal{KK}= (B_\mathcal{JJ}-B_{\mathcal{JI}}B_\mathcal{II}^{-1}B_{\mathcal{IJ}})_\mathcal{KK}=B_\mathcal{KK}-B_{\mathcal{KI}}B_\mathcal{II}^{-1}B_{\mathcal{IK}}=B_\mathcal{I'I'}/B_\mathcal{II}$. Thus any principal matrix of a Schur complement can be represented as a Schur complement. Then, by the determinant formula,
\[
   \det(B/B_\mathcal{II})_\mathcal{KK}= \det(B_\mathcal{I'I'}/B_\mathcal{II})= \det B_\mathcal{I'I'}/\det B_\mathcal{II} \geq 0.
\]
\end{proof}

We will also make use of the following characterization of $P_0$ matrices, observed by Fiedler and Pt\'ak \cite{fiedler1966}.
\begin{theorem}{(Fiedler \& Pt\'ak, 1966)}\label{t: Fiedler}
    $A\in \R^{N\times N}$ is a $P_0$ matrix if and only if for any nonzero $x\in \R^N$, there exists an index $i$ such that $x_i\not=0$ and $x_i(Ax)_i \geq 0$.
\end{theorem}

\begin{proof}[Proof of Theorem~\ref{t:Dexclusion}]
Without loss of generality, we assume that $\qt$ is an internally $D$-stable equilibrium with full support 
$\calQ=\Nset$, meaning $-B$ is $D$-stable. Suppose also that $\pt$ is a strictly stable equilibrium with
smaller support $\calI$. 
Now we can analyze the external stability condition \eqref{d:external} for $\pt$ using $B$ and $\qt$ as follows. 
For notational simplicity we let $\mathcal{J}=\mathcal{I}^c=\Nset\setminus\calI$.
We have 
\begin{align}
        a-B\pt &=
            \left( \begin{array}{c}
        a_\mathcal{I} \\
        a_\mathcal{J}
        \end{array} \right) -
        \left( \begin{array}{cc}
        B_\mathcal{II} & B_\mathcal{IJ} \\ 
        B_\mathcal{JI} & B_\mathcal{JJ}
        \end{array} \right)
        \left( \begin{array}{c}
        \pt_\mathcal{I}\\
        0
        \end{array} \right) 
        = \left( \begin{array}{c}
        0\\
        a_\mathcal{J} - B_{\mathcal{JI}}\pt_\mathcal{I}
        \end{array} \right),
        \end{align}
hence $a_\calJ-\BJI\pt_\calI = a_\calJ-\BJI\BII\inv a_\calI$.  Since $a=B\qt$, however, we can write  
\[
a_\calJ = \BJI\qt_I+\BJJ\qt_J, \qquad 
a_\calI = \BII\qt_I+\BIJ\qt_J, 
\]
and deduce from the external stability condition \eqref{d:external} that, componentwise,
        \begin{align}\label{e:extstable}
        0>   a_\mathcal{J} - B_{\mathcal{JI}}\pt_\mathcal{I} &= (B_\mathcal{JJ} - B_{\mathcal{JI}}B_\mathcal{II}^{-1}B_{\mathcal{IJ}}) \qt_\mathcal{J}
= (B/B_\mathcal{II}) \qt_\mathcal{J} ,
        \end{align}
where $B/B_\mathcal{II}$ is the Schur complement of $B_\mathcal{II}$ in $B$.
But since $B/B_\mathcal{II}$ inherits the $P_0$ property from $B$, this contradicts Theorem~\ref{t: Fiedler}.
\end{proof}

This argument yields a result that differs in a rather subtle way from 
conclusions implied by Theorem 15.4.5 in the book of Hofbauer and Sigmund  \cite{HofbauerSigmund}. 
This theorem states that $B$ is a $P$ matrix (meaning all its principal minors are positive)
if and only if \emph{for every $a\in\R^N$},
the system \eqref{eq:generalLV} has a unique equilibrium  $\pt\in\RNge$
which is ``saturated,'' meaning $a_i-(B\pt)_i\le0$ for all $i$.
Any strictly stable equilibrium is strictly saturated, so  
it follows from \cite[Thm.~15.4.5]{HofbauerSigmund} 
that if $\calQ\subseteq\Nset$ and $\BQQ$ is a $P$-matrix, 
then at most one subcommunity of $\calQ$  can be $\calQ$-stable.

The same proof as that of Theorem~\ref{t:Dexclusion} above establishes the following 
related exclusion principle, which relaxes the assumption on strict positivity of minors
while strengthening the saturation (exterior stability) condition.
If $\calI\subseteq\calQ\subseteq\Nset$, let us call an equilibrium $\pt$ with support $\calI$
\emph{strictly $\calQ$-saturated} if  $\BII$ is nonsingular and $a_i-(B\pt)_i<0$ for all $i\in\calQ\setminus\calI$.

\begin{corollary}\label{c:ssaturated}
If a community $\calQ\subseteq\Nset$ supports an equilibrium $\qt\in\RNge$ and
$\BQQ$ is a $P_0$ matrix, then no different equilibrium $\pt$ supported inside $\calQ$ can be 
strictly $Q$-saturated. In particular, no $\pt\ne\qt$ can be $\calQ$-stable.
\end{corollary}

The internal $D$-stability condition in Theorem~\ref{t:Dexclusion} 
is in principle weaker than the VL-stability condition in Theorem~\ref{t:VLexclusion}.
As we have indicated, it is not known how to verify $D$-stability computationally in every case where it is true, when $N>4$.
In contrast, the assumptions in both Corollary \ref{c:ssaturated} and \cite[Thm.~15.4.5]{HofbauerSigmund} can in principle be checked by computing sufficiently many principal minors. 
In practice, though, the number of minors involved may become prohibitively large if many species are considered.

\subsection{Counterexamples to  exclusion in competitive systems}\label{ss:counterex}
Equations \eqref{eq:generalLV} model purely competitive interactions if all entries of the matrix $B$ are positive.
Lischke and L\"offler \cite{lischke2017finding} state that in their extensive simulations of random competitive Lotka-Volterra systems, they never encountered a case where a single species formed a stable subcommunity of a larger one. 
Example~\ref{ex:fail ex} shows that this is not a universal property that holds for all competitive systems,
but the results of \cite{lischke2017finding} suggest that encountering counterexamples may be a rare event.
In this subsection we show that one can invent such counterexamples in systems of any size $N\ge3.$

For a single-species community $\calI$ to be stable and contained in a larger one $\calQ$, 
necessarily $-\BII<0$, and $-\BQQ$ must be stable but not $D$-stable. 
For definiteness  we set $\calI=\{1\}$, $\calQ=\Nset$ and $\calJ=\calQ\setminus\calI$.

A matrix $B=\BQQ$, with Schur complement $C=B/\BII$,  
might have these properties if $B_{11}>0$ and $C$ has some negative diagonal element 
(implying $C$ is not a $P_0$ matrix).
We can seek $B$ in the block form
\begin{equation}
    B = \begin{pmatrix} b & r^T \\ c & C+cr^T/b \end{pmatrix},
\end{equation}
where $b>0$ and $c,r\in\R^{N-1}$ have positive entries. (Note $C=B/\BII$ here.)

In order for a state  $\qt=(\qt_1,\qt_\calJ)^T>0$ 
to be a strictly stable equilibrium,
we require $a=B\qt$ and all eigenvalues of $\dg{\qt} B$ to have positive real parts. 
Then in order for $\pt=(\pt_1,0)^T$ to be strictly stable, it suffices 
by \eqref{e:extstable} that 
\begin{equation}\label{e:psstable}
b\pt_1=a_1=b\qt_1+r^T\qt_\calJ 
\quad\text{and}\quad
a_\calJ -\BJI \pt_\calI = C \qt_\calJ <0.
\end{equation}

In Example~\ref{ex:fail ex} these conditions all hold --- e.g.,
\[
C\qt_\calJ= \begin{pmatrix}-1 &1\\-2&1\end{pmatrix}
 \begin{pmatrix}2 \\1\end{pmatrix}<0.
\]
To construct examples for any $N\ge3$, it is convenient to choose $C$ to make $B$ a rank-2 perturbation of $\eps I$ for small $\eps>0$. That is, we seek to make 
\begin{equation}
B=\eps I + vw^T+\hat v\hat w^T, \quad 
\end{equation}
where the vectors $v,w,\hat v,\hat w\in\R^N$ have the block form
\[
v=\pmat{1\\c}, \quad
w=\pmat{1\\r}, \quad
\hat v=\pmat{0\\-\hat c}, \quad
\hat w=\pmat{0\\\hat r}. 
\]
In this case $b=1+\eps$ and the Schur complement  
\[
  C=-\hat c\hat r^T + \eps(I+cr^T/b).
\]
The matrix $B$ has the eigenvalue $\eps>0$ with multiplicity $N-2$, since any vector orthogonal to both $w$ and $\hat w$
is an eigenvector. It is straightforward to show that the two remaining eigenvalues
must take the form $\eps +\lambda$ where $\lambda$ is an eigenvalue of the $2\times2$ matrix
\begin{equation}\label{d:Mmat}
M = \pmat{w^Tv &w^T\hat v \\ \hat w^T v & \hat w^T\hat v} .
\end{equation}
With the specific choices 
\[
r = (1,1,\ldots,1)^T,
\quad c = 3r,
\quad \hat r = (1,0,\ldots,0)^T, 
\quad \hat c = 2r-\hat r ,
\]
we find 
\[
M = \pmat{1+3m & 1-2m \\ 3 & -1 }, \qquad m=N-1.
\]
The eigenvalues of $M$ have positive real part for all  $N\ge3$,
since then $M$ has positive trace $3m$ and determinant $3m-4$.

Thus the matrix $-B$ is stable.
With the choices $\qt=(1,1,\ldots,1)^T$, $a=B\qt$, the state
$\qt$ becomes a strictly stable equilibrium.  
With $\pt_1=a_1/(1+\eps)$, the state 
$\pt=(\pt_1,0)^T$ then satisfies \eqref{e:psstable} for sufficiently small 
$\eps>0$, since 
\[
C\qt_\calJ = -\hat c + O(\eps) = (-1,-2,\ldots,-2)^T + O(\eps)<0.
\]
Thus the single-species equilibrium  $\pt$ is also strictly stable for small $\eps>0$.

\section{Bounds from Sperner's lemma}\label{ss:sperner}

The exclusion principles of the previous section  imply bounds on the number of 
stable communities of certain types,
by a well-known result from the combinatorial theory of \emph{posets} (partially ordered sets).
 A {poset} is a set $P$ with a binary relation $\leq$ satisfying reflexivity ($a\leq a$), antisymmetry (if $a\leq b$ and $b\leq a$, then $a=b$) and transitivity (if $a\leq b$ and $b\leq c$, then $a\leq c$).
Two elements $a$ and $b$ in $P$ are \emph{comparable} if $a\leq b$ or $b\leq a$.
        A \emph{chain} in $P$ is a subset $\calC\subseteq P$ such that any two elements in $\calC$ are comparable.
        \begin{definition}
        An {\em anti-chain} in a poset $P$ is a subset $\calA \subseteq P$ such that no two elements in $\calA$ are comparable.
    \end{definition}
    
For any set $S$, the collection of all subsets of $S$ ordered by inclusion is a poset, 
denoted by $(\calP(S),\subseteq)$. 
For $S=\{1,2,3\}$, e.g.,
        the collection $\big\{\emptyset, \{1\}, \{1,2\}, \{1,2,3\}\big\}$ is a chain and $\big\{\{1,2\},\{2,3\},\{1,3\} \big\}$ is an anti-chain.
The maximal size of any anti-chain in a finite poset is bounded by the following well-known result of Sperner. See \cite{lubell1966} for a short proof.
    \begin{lemma}[Sperner's lemma]
        Let $\calA$ be an anti-chain in a poset $P$ having N elements. Then the number of elements of $\calA$ is at most  $\binom{N}{\lfloor N/2\rfloor}$.
    \end{lemma}
    
From Theorem \ref{t:Dexclusion} we directly infer the following.

\begin{corollary}\label{c:ds bounds}
For any Lotka-Volterra system \eqref{eq:generalLV}, no two stable subcommunities of $\Nset=\{1,2,\ldots,N\}$ that are internally $D$-stable are comparable with respect to inclusion.
The number of strictly stable equilibria that are internally $D$-stable is therefore at most 
$\binom{N}{\lfloor N/2\rfloor}$.
    \end{corollary}
    
\begin{remark}
Note that if $B$ is $D$-symmetrizable, any strictly stable state is internally $D$-stable. In this case the number of strictly stable equilibria is bounded above by $\binom{N}{\lfloor N/2\rfloor}$.
\end{remark}

We remark that when $N$ is large, this bound is exponentially large in $N$ and not so very much smaller than $2^N$, the number of all subsets of $\Nset$. 
For Stirling's approximation says $n!\sim\sqrt{2\pi n}(\frac{n}{e})^n$, thus
    \begin{equation}\label{e:stirling}
     \binom{N}{\lfloor N/2\rfloor}\sim 2^N \sqrt{\frac{2}{\pi N}}.
    \end{equation}

\begin{remark}
The same type of anti-chain property as described for Lotka-Volterra systems in Corollary~\ref{c:ds bounds} 
is well-known to hold for the supports of evolutionarily stable states (ESSs) in evolutionary game theory.
The Bishop-Cannings theorem~\cite[Thm.~2]{bishop1976models} implies that
the support of any ESS can neither contain nor be contained in the support of any other.
This theorem about ESSs actually provides a different collection of Lotka-Volterra communities that enjoy
the anti-chain property.  We discuss this in detail below in Section~\ref{s:ESS}.
\end{remark}

\section{Multiplicity of stable steady states}\label{s:multiplicity}

We do not know whether the bound in Corollary~\ref{c:ds bounds} that comes from Sperner's lemma is sharp.
For certain systems whose interactions have a bimodal competition structure, though,
the number of stable communities can be exponentially large in $N$, 
and greater than $2^{N/2}$ in particular.  
This number is a bit larger than the square root of the bound in \eqref{e:stirling}.

Systems with such great numbers of stable communities may be quite rare. 
In the course of extensive numerical explorations of a random class of Lotka-Volterra systems, Lischke and L\"offler \cite[Table 2]{lischke2017finding} found that multiple stable equilibria occur in about half of their simulations involving between 2 and 60 species, with about 2 percent having more than 25 stable equilibria. 
It appears that no more than about 40 stable equilibria in one system were ever encountered
in \cite{lischke2017finding}. 
With $N=60$, though, more than $2^{N/2}>10^9$ stable equilibria are possible in theory. Thus we are interested to  
investigate whether robust conditions can be described which ensure that large numbers of stable equilibria exist.

\subsection{Indistinguishable species}
One property that can allow many stable communities to exist is that 
stability persists if some species in a community is exchanged for a different species. 
If such a stability-preserving exchange is possible for $m$
different pairs of species independently, then the number of stable communities
is at least as large as $2^m$.  

The simplest type of exchange of this kind occurs for two species 
$x$ and $y$ with identical growth rates and interaction coefficients, satisfying 
\begin{equation}
a_x=a_y, \quad B_{ix}=B_{iy}, \quad B_{xj}=B_{yj},
\end{equation}
for all $i,j\in\Nset$.  We call $x$ and $y$ {\it indistinguishable} in this case.

For two such species, permuting the index labels in the Lotka-Volterra system~\eqref{eq:generalLV}
by swapping $x$ and $y$ leaves the system invariant.
Thus if $\calI$ is a stable community that contains $x$ but not $y$,
then the community $\hat\calI$ obtained by replacing $x$ by $y$ is also stable.

We shall describe two examples which involve groups of indistinguishable species,
permitting large numbers of stable communities.

\begin{exmp}[Complete indistinguishability and competitive exclusion] In the simplest case, 
all species are pairwise indistinguishable,
with interspecific competition coefficients all the same, 
and intraspecies competition coefficients also all the same: 
\begin{equation}\label{ex:m1}
    B_{ij}=
    \begin{cases}
    \alpha &\text{if $j\ne i$},\\
    \beta &\text{if $j=i$},
        \end{cases}
        \qquad 
        a_i = \beta.
    \end{equation}
When $\alpha>\beta>0$ 
this system has {exactly $N$ strictly stable 
steady states} $\pt$ with $\pt=\one_{\calI}$ for any singleton set $\calI=\{i\}$, $i=1,\ldots,N$.
Because $B$ is symmetric, Corollary~\ref{cor:LL} applies. 
Thus, when the interspecific competitions are stronger than the intraspecific competition,  
the competitive exclusion principle is valid. 
(When $0<\alpha<\beta$ on the other hand, $B$ is positive definite and the system has a unique strictly stable equilibrium  having equal population densities for all species.)
\end{exmp}

A much larger number of stable communities can be obtained. 
Suppose the set of $N$ species can be partitioned into $m$ disjoint subsets,
each of which consists of pairwise indistinguishable species,
and suppose further that a stable community $\calI$ exists 
that contains \emph{exactly one member from each subset}.
Then each member of $\calI$ can be exchanged with any member indistinguishable from itself.
If $k_1,k_2,\ldots,k_m$ denote the number of species in the $m$ different subsets,
then the number of stable communities in this case is at least as large as the number
\begin{equation}\label{n:kprod0}
\prod_{j=1}^m k_j = k_1k_2\cdots k_m.
\end{equation}
We will show that this is indeed possible for any partition of $N$,
as a special case of the main result in the next subsection.
See Example~\ref{x:rivals} below.

\subsection{Weak vs strong competition}
As mentioned in the Introduction, some recent biological studies
suggest that weak interactions may predominate in certain naturally 
occurring microbiomes, but stability is enhanced by the presence of 
some strongly competitive interactions.
In this section we describe examples with this nature,
having many stable communities.  

Our construction is motivated by a known result in evolutionary
game theory for symmetric payoff matrices related to the incidence
matrix of a general graph.  For such matrices,
Cannings and Vickers \cite[II]{vickers1988patterns} state that
the ESSs are characterized in terms of the \emph{cliques} of the graph (maximal complete subgraphs).
In the context of continuous-time models of allele selection in population
genetics, with a symmetric fitness matrix of this type,
Hofbauer and Sigmund \cite[Sec.~19.3]{HofbauerSigmund}
state that all stable rest points are characterized in terms of the cliques.
Below we prove that a result of this type holds for Lotka-Volterra systems.

\begin{exmp}[Friends vs rivals] \label{x:friends}
We suppose that any two different species $i$ and $j$ are either 
relatively \emph{friendly} or are strong \emph{rivals}.
The interspecific interaction coefficients $B_{ij}$ 
will take only three values: $\alpha$ (modeling friendly competition), 
$\beta$ (self-inhibition), and $\gamma$ (strong rivalry), and we assume
\begin{equation}\label{c:abc}
\alpha < \beta < \gamma.
\end{equation}
We set $a=\one$ ($a_i=1$ for all $i$) and
\begin{equation} \label{e:Brival}
B_{ij} = \begin{cases}
\alpha  & \text{if $i$ and $j$ are friendly},\\
\beta & \text{if $i=j$},\\
\gamma & \text{if $i$ and $j$ are rivals.}
\end{cases}
\end{equation}
Evidently $B$ is symmetric.  If $\alpha=-1$ and $\beta=\gamma=0$, the matrix
$-B$ is the incidence matrix for the graph whose edges 
connect friendly species.

In this context, a \emph{clique} is a maximal set of mutually friendly species. 
That is, a set $\calI\subset\Nset$ is a clique if  every two different species $i,j\in\calI$ 
are friendly, and no species $k\notin\calI$ is friendly with all the species in $\calI$
(so every $k\notin\calI$ has some rival in $\calI$).

Under the assumptions above, in this example we have the following.
\begin{prop}\label{p:clique}
Let $\calI\subset\Nset$ be a set with $m$ members. 
Then $\calI$ is a stable community if and only if $\calI$ is a clique and
\begin{equation}\label{d:cm}
    c_m := \beta+(m-1)\alpha>0. 
\end{equation}
\end{prop}
As a corollary, the stable communities
in this example coincide exactly with the cliques, 
provided we know $c_m>0$ for every clique.
This holds in particular if $\alpha\ge0$, meaning all interactions in the system are competitive.
If $\alpha<0$, it holds if $\beta>(M-1)|\alpha|$, where $M$ is the size of the largest clique.
\begin{proof}
First, suppose $\calI$ is a clique and \eqref{d:cm} holds. Then the state 
\begin{equation}\label{d:pcliq}
p=\frac{\one_\calI}{\beta+(m-1)\alpha}, \qquad\mbox{with}\quad
p_i = \begin{cases} 1/c_m & \mbox{if $i\in\calI$,}\\0 & \mbox{otherwise,}
\end{cases}
\end{equation}
is an equilibrium,
and the matrix $A$ in the linearized equation \eqref{e:linear} has the following structure:
Whenever $i\notin\calI$ we have $A_{ij}=0$ for all $j\ne i$, and moreover,
because $i$ has at least one rival $j\in\calI$ and $A_{ii}=a_i-(Bp)_i$,
\begin{align}\label{e:cajj}
   c_m A_{ii} &= c_m-\sum_{j\in\calI}B_{ij} \le c_m-\gamma-(m-1)\alpha = \beta-\gamma<0.
\end{align}
This is the external stability condition. On the other hand, because 
the block $\AII = -\dg{p_\calI} \BII$ and $c_m\dg{p_\calI} = I$, we find that
\begin{equation}\label{e:Bba}
-c_m\AII = \BII = (\beta-\alpha)I + \alpha \one\one^T .
\end{equation}
The eigenvalues of this symmetric matrix are $c_m$ (with eigenvector $\one$) and 
$\beta-\alpha$ (with eigenspace orthogonal to $\one$). Since both are positive,
$\AII$ is negative definite. Hence $p$ is strictly stable, so $\calI$ is a stable community. 

Conversely, suppose $\calI$ is a stable community, supporting
a strictly stable equilibrium $p$.  Necessarily $\AII=-\dg{p_\calI}\BII$ is stable,
and so also is the similar (and symmetric) matrix $-\dg{q} \BII\dg{q}$ where 
$q_i=\sqrt{p_i}$ for all $i$.
By Sylvester's law of inertia, $\BII$ is necessarily positive definite.
Then it follows that $\calI$ contains no pair of rivals, for otherwise the indefinite matrix
\[
\begin{pmatrix}
    \beta & \gamma \\ \gamma & \beta
\end{pmatrix}
\]
would be a principal submatrix of $\BII$. 

Thus $\calI$ is a set of mutually friendly species, and necessarily $\BII$ has the form in \eqref{e:Bba}.
It follows that the eigenvalue $c_m>0$ and that $p$ takes the form in \eqref{d:pcliq}. 
If $\calI$ is not itself a clique,
then some $i\notin\calI$ is friendly with all $j\in\calI$, and as in \eqref{e:cajj}
we calculate that  $c_mA_{ii}=\beta-\alpha>0$. This contradicts the strict stability of $p$. 
Hence $\calI$ is a clique.
\end{proof}
\end{exmp}

\begin{remark}
In the example above, the stability of a given community $\calI$ of
mutually friendly species persists under a slight loosening of the constraints 
on the interspecific interaction coefficients.  
Namely, one need not assume the symmetry $B_{ij}=B_{ji}$ for species $i\notin\calI$. 
It is only necessary that each such species $i$ be strongly
inhibited by some member $j\in\calI$, having $B_{ij}>\beta$ 
(for this ensures $A_{ii}<0$ in \eqref{e:cajj}).
No condition regarding the inhibition $B_{ji}$ of species $j$ by $i$ is needed.
\end{remark}

In general it does not seem quite easy to count all the cliques in a graph,
so we describe a class of special cases which shows that the number of stable communities
in \eqref{n:kprod0} can be achieved (cf.~\cite[Exercise~19.3.3]{HofbauerSigmund}).

\begin{exmp}[Partitioning by rivals] \label{x:rivals}
Suppose that in the preceding example, the $N$ species can be partitioned into $m$ disjoint and nonempty sets 
of \emph{mutual rivals}, respectively having $k_1,k_2,\ldots,k_m$ members,
and any two species from different sets are friendly.
Then clearly each clique (maximal set of mutually friendly species) has $m$ members and
is composed of one member from each set of rivals. Moreover, the number of cliques
is given by \eqref{n:kprod0}. Provided \eqref{d:cm} holds for this value of $m$,
these cliques comprise  all the possible stable communities.
\end{exmp}

The maximum number of cliques in a graph of $N$ nodes is \cite{moon1965} 
\begin{equation}
n\cdot 3^{m-1}, \quad\mbox{if $N=3(m-1)+n$ with $n=2,3$ or $4$} .
\end{equation}
This is therefore the maximum number of ESSs occurring in the main example considered in  \cite[II]{vickers1988patterns}.
In Example~\ref{x:rivals} we achieve this number with $m-1$ sets of 3 rivals each and one set of $n$.
For $N=60$ we have $m=20$ and find $3^{20}\approx 3.49\times 10^9$ strictly stable equilibria 
can occur in such a system.

\subsection{Robust criteria for stability of cliques} 
The property of being a stable community naturally persists under 
\emph{sufficiently small} changes in the growth rates $a_i$ and 
interaction coefficients $B_{ij}$.  
But the mathematical notion of ``sufficiently small'' leaves it unclear
just \emph{how} small a change is allowed.
Here we aim to describe a simple and explicit set of 
quantitative bounds which ensure that a community $\calI$ is stable,
focusing on cases qualitatively similar to Example~\ref{x:friends},
in which $\calI$ essentially consists of a maximal set of mutually friendly species.

Recall that, for given $a$ and $B$, a community $\calI$ is stable 
if it supports a strictly stable equilibrium $p$. This means exactly that,
in the notation of section~2.2, the following conditions hold:
\begin{itemize}
\item[(i)] For all $i\in\calI$, $a_i=\sum_{j\in\calI}B_{ij}p_j$ and $p_i>0$.
\item[(ii)] For all $i\notin \calI$, $a_i<\sum_{j\in\calI}B_{ij}p_j$ and $p_i=0$.
\item[(iii)] $\BII$ is nonsingular and $\AII =-\dg{p_\calI}\BII$ is stable. 
\end{itemize}
For any specific case, general perturbation results for 
linear systems~\cite[Sec.~2.7]{GolubVanLoan}
and matrix stability~\cite[Thm.~2.4]{hewer1988}
can be invoked to provide quantitative bounds for
changes in $a$ and $B$  which ensure that these properties 
persist for a perturbed equilibrium with the same support.

We do not develop such results here, but instead pursue 
the limited aim of describing a set of systems 
in which interspecific competition is 
bimodal---either weak or strong---and that are qualitatively similar 
to Example~\ref{x:friends},
having multiple stable communities formed by cliques.

For simplicity, we will consider only competitive systems
for which 
\begin{equation}\label{c:compete}
a_i>0 \quad\mbox{and}\quad  B_{ij}\ge 0 \quad\mbox{ for all $i,j\in\Nset$.}
\end{equation}
For notational convenience we also suppose that a diagonal scaling 
as in \eqref{e:Dscale} has been performed with $d_{jj}=a_j/B_{jj}$, corresponding to  
\begin{equation}\label{c:abii}
\hat B_{ij} = \frac{B_{ij}a_j}{B_{jj}}, \qquad \hat p_i = \frac{B_{ii}p_i}{a_i}
\qquad\mbox{for all $i,j\in\Nset$,}
\end{equation}
whence $a_i = \hat B_{ii}$ for all $i$. 
\begin{prop} 
Assume \eqref{c:compete} and let $\alpha\in(0,\frac12)$.
Suppose $\calC$ is some collection of communities $\calI\subset\Nset$
for which the following hold:
\begin{align}
 \sum_{j\in\calI,j\ne i}\hat B_{ij} &\le \alpha\, a_i
\quad\mbox{for each $i\in\calI$,}
\label{C:friend}
\\  \sum_{j\in\calI} \hat B_{ij} &> \frac{a_i}{1-\alpha}
\quad\mbox{for each $i\notin\calI$.} 
\label{C:rival}
\end{align}
Then each $\calI\in\calC$ is a community that supports a strictly
stable equilibrium which globally attracts all solutions having the same support.
\end{prop}
Evidently, condition \eqref{C:friend} requires that  for species within $\calI$, 
the (total) interspecific competition  is weak compared to self-inhibition,
and \eqref{C:rival} requires that each species not in $\calI$  
is strongly competed against (in total) by the species inside $\calI$.  

\begin{proof}
Let $\calI\in\calC$ have $m$ members.   
A state $\hat p\in\RNgt$ supported by $\calI$ is an equilibrium for \eqref{e:Dscale}
if and only if 
\begin{equation}\label{e:fixp}
\hat p_i = F(\hat p)_i:= 1 - \frac1{a_i}\sum_{j\in\calI,j\ne i} \hat B_{ij}\hat p_j 
\quad\mbox{for all $i\in\calI$}.
\end{equation}
Under the given hypotheses, the function $F$ is a strict contraction in the max norm
on $\R^m$ given by $\|v\|_\infty = \max_{i\in\calI}|v_i|$, since
\[
\|F(v)-F(w)\|_\infty \le \alpha\|v-w\|_\infty .
\]
The set $S=[1-\alpha,1]^m \subset\R^m$ is mapped into itself by $F$, 
hence $F$ has a unique fixed point in $S$ given by $\hat p_\calI$,
where $\hat p$ is an equilibrium of \eqref{e:Dscale} supported by $\calI$ satisfying  
$1-\alpha\le  \hat p_i\le 1$ for each $i\in\calI$.
Condition \eqref{C:rival} ensures that for each $i\notin\calI$, 
\[
a_i - \sum_{j\in\calI} \hat B_{ij}\hat p_j \le a_i - 
(1-\alpha)\sum_{j\in\calI} \hat B_{ij} <0.
\]
Hence conditions (i) and (ii) above for a strictly stable equilibrium hold. 
Condition (iii) holds also because the matrix $C=\dg{\hat p_\calI}\hat B_{\calI\calI}$
(similar to $-\AII$) is diagonally dominant: 
Indeed, for all $i\in\calI$ we have 
\[
C_{ii} = \hat p_i \hat B_{ii} = a_i - \sum_{j\in\calI,j\ne i} \hat B_{ij} \hat p_j \ge a_i(1-\alpha)
\]
since \eqref{c:abii} and \eqref{e:fixp} hold and $ \hat p_j\le1$, while 
\[
\sum_{j\in\calI,j\ne i}|C_{ij}|  =  \hat p_i\sum_{j\in\calI,j\ne i}  \hat B_{ij} \le a_i\alpha<a_i(1-\alpha).
\]
By Gershgorin's theorem, every eigenvalue of $C$ has positive real part.
Moreover, the diagonal dominance of $C$ also implies $\AII$ is VL-stable
(by \cite[Thm.~3]{moylan1977}, or see the remark below).
Hence the community $\calI$ is internally VL-stable, and the equilibrium $p$ that 
it supports globally attracts all solutions with the same support.
\end{proof}
\begin{remark}\label{r:d2vl}
We sketch a proof  that $-C$ is VL-stable (cf.~\cite{tartar1971}) 
for the reader's convenience. 
Let $G_{ij}=|C_{ij}|/C_{ii}$ for $i\ne j$, and $G_{ii}=0$. 
Then $I-G$ is a diagonally dominant $M$-matrix,
with inverse $\sum_{k\ge0}G^k$ whose entries are all nonnegative.
Hence $q=(I-G)^{-T}\one \in\RNgt$, and 
it follows that $C^T\dg{q}$ is diagonally dominant, for
$
C_{ii}q_i - \sum_{j\ne i} |C_{ji}q_j| = C_{ii}>0.
$
Because $\dg{q}C$ is diagonally dominant too, $DC+C^TD>0$ where $D=\dg{q}$.
\end{remark}

\section{Relation to evolutionary game theory}\label{s:ESS}

In evolutionary game theory, there is a substantial body of research on
multiplicity and patterns of \emph{evolutionarily stable states} (ESSs)
and the dynamics of \emph{replicator equations},
which bears a close comparison with the results we have developed in this paper for Lotka-Volterra systems.
For various known facts about these things that we mention below, we refer to the 
books of Hofbauer and Sigmund \cite{HofbauerSigmund} and Hadeler \cite[Sec.~3.4]{hadeler2017book}.

{\it Correspondence.} The dynamics of the Lotka-Volterra system \eqref{eq:generalLV} in $\RNge$  
is well-known to correspond to those of \emph{replicator equations} of the form
\begin{equation}\label{e:rep}
    x'_i = x_i((Ax)_i-x^T Ax) , \quad i=0,1,\ldots,N,
\end{equation}
with $x$ in the $N$-simplex 
$\Delta_N$ consisting of all $x=(x_0,x_1,\ldots,x_N)$ such that $x_i\ge0$ for all $i$ and 
$\sum_{i=0}^N x_1=1$,
via the mapping $p\mapsto x$ given by
\begin{align}\label{e:lv2rep}
    x_{0} = 1/(1+\sum_{j=1}^N p_j),\qquad
    x_i &= p_i/(1+\sum_{j=1}^N p_j), \quad i=1,\cdots, N. 
\end{align}
This works for the payoff matrix
\begin{equation}\label{e:Aform}
A=\begin{pmatrix} 0&0\\ a & -B \end{pmatrix},
\end{equation}
and after a solution-dependent nonlinear change of time variable.

{\it Notion of ESS.} An important notion in evolutionary game theory is 
the following: 
\begin{definition}\label{d1:ess}
A state
$y\in\Delta_N$ is an {\it evolutionarily stable state (ESS)} when the following conditions are satisfied:
\begin{itemize}
    \item[(a)] $y^T Ay \geq x^T A y$, for all $x\in \Delta_N$
    \item[(b)] if $x\ne y$ and $y^T Ay =x^T Ay$ then $y^T Ax > x^T Ax$, for all $x\in\Delta_N$
\end{itemize}
\end{definition}
An equivalent characterization is that $y$ is an ESS if and only if 
\begin{equation}\label{d2:ess}
y^TAx>x^TAx \quad\mbox{ for all $x\ne y$ near enough to $y$ in $\Delta_N$.}
\end{equation}
The condition (a) alone makes $y$ a \emph{Nash equilibrium}.
It is known that any ESS is a steady state that is locally attracting 
(nonlinearly asymptotically stable) for replicator dynamics.
If $A$ is symmetric, any steady state is locally attracting if and only if it is an ESS.
If $A$ is not symmetric, however, a locally attracting steady state of \eqref{e:rep} 
need not be an ESS.

In what follows, we will describe conditions that characterize 
Lotka-Volterra equilibria that correspond to ESSs in the way above.
Our goal is to describe what stability properties such ESS-derived equilibria
must or may not have, and compare known exclusion principles for ESSs
to those we have developed in this paper. 

{\it Symmetries.}
A few relevant facts are the following:
The correspondence holds and the mapping $p\mapsto x$ can be reversed under the \emph{proviso} that $x_0\ne0$.
Replicator dynamics are known to be invariant under two kinds of
transformations, one that modifies all entries in any column of $A$ by adding a
constant $b_i$, and one that scales by a positive diagonal matrix $D$:
\begin{itemize}
\item[(i)] $A\mapsto A+\one b^T$ and  $x\mapsto x$ with the same time scale,
\item[(ii)] $A \mapsto AD$ and $x\mapsto D\inv x/(\one^TD\inv x)$ with a nonlinear time change.
\end{itemize}
Using a transformation of type (i), one can map any replicator equation 
in $\Delta_N$ with $x_0\ne0$ to an $N$-component Lotka-Volterra system. 
We note, however, that these correspondences do not generally 
allow symmetric $A$ to correspond with symmetric $B$ in \eqref{eq:generalLV} or vice versa.

Meanwhile, recall from \eqref{e:Dscale} that Lotka-Volterra systems are invariant under a positive diagonal scaling on $B$:
\begin{equation}
\mbox{$B\mapsto BD$ and $p\mapsto D^{-1}p $ with the same time scale.}
\end{equation}
One can expect that internal stability (see Def.~\ref{d:internal}) of equilibria of Lotka-Volterra systems will be conserved through the transformation, and it is true indeed. In replicator equations, however, a transformation of type (ii) can disrupt an ESS. In other words, when  $y$ is an ESS, an image $\hat y=D\inv y/(\one^TD\inv y)$, under a transformation of type (ii), is a Nash equilibrium but might \emph{not} be an ESS. We provide an example regarding this issue below.
 
{\it Relation to Lotka-Volterra.}
If $y=(y_0,y_1,\ldots,y_N)$ is an ESS with $y_0>0$, 
it corresponds to a locally attracting steady state 
\begin{equation}\label{d:y2q}
q = y_0\inv(y_1,\ldots,y_N) 
\end{equation}
for the Lotka-Volterra system \eqref{eq:generalLV} obtained 
by reducing $A$ to the form \eqref{e:Aform} by a transformation of type (i) above.
One can readily check that the Nash equilibrium condition (a) corresponds to the condition that,
for all $i=1,\ldots,N$,
\begin{equation}\label{c:saturated}
(a-Bq)_i=0 \quad \mbox{if $q_i>0$},
\qquad (a-Bq)_i\le0 \quad \mbox{if $q_i=0$}.
\end{equation}
A state $q$ satisfying these conditions is called a \emph{saturated fixed point} in \cite{HofbauerSigmund}.

The ESS condition \eqref{d2:ess} translates to mean that
\[
\left(\frac{1+\one^T p}{1+\one^T q} q - p\right)^T(a-Bp) >0 \quad\mbox{for all $p\ne q$ near $q$ in $\RNge$.}
\]
Substituting $p=q+r$, this is equivalent to saying that for all small enough $r$ with $q+r\in\RNge$,
\begin{equation}\label{e:rversion}
0< \left(\left(I-\frac{q\one^T}{1+\one^Tq}\right)r\right)^T (Br+Bq-a).
\end{equation}
Substituting $r=(I+q\one^T)v$, one then finds the following characterization.
\begin{lemma}\label{l:lvess}
A state $y\in\Delta_N$ with $y_0>0$ is an ESS for $A$ in the form \eqref{e:Aform}
if and only if for all nonzero $v\in\R^N$ small enough we have
\begin{equation}\label{e:vversion}
0< v^TB(I+q\one^T)v + v^T(Bq-a) \quad\mbox{if $v_i\ge0$ whenever $q_i=0$.}
\end{equation}
\end{lemma}

From this characterization we can infer the following.
If $q_i>0$ for all $i$, then $Bq=a$ and it is necessary and sufficient 
for $y$ to be an ESS that the symmetric part of $B+a\one^T$ is positive definite.
(Or equivalently, the symmetric part of $B+\one a^T$ is positive definite.)

In general, if $q_i=0$ for some $i$, let $\calI=\spt q$, then 
since $(Bq-a)_\calI=0$, necessarily
\begin{equation}
0< v_\calI^T (\BII+a_\calI\one_\calI^T)v_\calI \quad\mbox{for all nonzero $v\in\R^N$},
\end{equation}
meaning the symmetric part of $\BII+a_\calI\one_\calI^T$ is positive definite.
This implies the symmetric part of $\BII$ is positive definite on the block subspaces 
of dimension $|\calI|$ orthogonal to both $\one_\calI$ and $a_\calI$.

It is natural to ask how \eqref{e:vversion} is related to internal stability in Lotka-Volterra equation. Considering that the ESS property brings nonlinear asymptotic stability, we cannot expect that $-\dg{q_{\calI}} \BII$ to be exponentially unstable. Combined with the external stability that the Nash condition provides, we have the following implication for the image of an ESS in the Lotka-Volterra system. 
\begin{theorem}\label{t:ESSinLV}
     Let $y$ and $q$ be equilibria of the replicator and Lotka-Volterra equations, respectively, that are equivalent in the sense of \eqref{d:y2q}. If $y$ is an ESS,  then $q$ is an internally stable, saturated fixed point. 
\end{theorem}
\begin{proof}
    We already checked the saturated condition in \eqref{c:saturated}. Note that in \eqref{e:rversion}, $\spt r\subseteq \spt q$ if and only if $\spt v\subseteq \spt q$. So without loss of generality, we can assume that $q$ is of full support.
    Let's define $C=(I+\one q^T)\inv B$ 
    and rewrite \eqref{e:rversion} as follows:
\[
0< r^TCr=r^T\left(I-\frac{\one q^T}{1+\one^Tq}\right)Br, \quad\mbox{for all nonzero $r\in\R^N$}.
\]
In order to check internal stability of $q$, let's examine 
\[
\dg{q}B=\dg{q}(I+\one q^T)C=(\dg{q}+q q^T)C.
\]
Since $\dg{q}+q q^T$ is (symmetric) positive definite, we can find a positive definite matrix $R$ such that $R^2=\dg{q}+q q^T$. Then,
\[
\dg{q}B\sim R^{-1}\dg{q}BR=RCR.
\]
Note $-RCR$ is Volterra-Lyapunov stable since $r^TRCRr=(Rr)^TC(Rr)>0$, which implies that $-RCR$ is stable. Thus by similarity of $\dg{q}B$ and $RCR$, we can conclude that $-\dg{q}B$ is stable, i.e., $q$ is internally stable. 
\end{proof}
{\it Relation to strict stability.}
We would like to point out that Theorem~\ref{t:ESSinLV} is sharp in the sense that the ESS property neither implies nor is implied by strict stability of the corresponding equilibrium in the Lotka-Volterra system. The following examples not only support this but also bring out the problematic lack of an intrinsic dynamical nature for the ESS property.
\begin{exmp}
For $N=2$, let $0<\alpha<2$, $\beta>1+\alpha$ and 
\[
B = \begin{pmatrix} 1 & -\alpha\\ \beta &0\end{pmatrix},  \qquad q = \one, \qquad a = Bq .
\]
Then $-B$ is $D$-stable, and strictly stable (but not VL-stable).
The definiteness condition in Lemma~\ref{l:lvess} fails to hold, however, since for $v^T=(1,-1)$ we have
\[
v^TB(I+q\one^T)v = 1+\alpha-\beta<0.
\]
Therefore the corresponding state $y=\frac13(1,1,1)\in\Delta_2$ for the replicator system is not an ESS
for the corresponding matrix $A$ in \eqref{e:Aform}.

However, if we consider
\[
\Tilde{B}=BD,  \qquad \Tilde{a}=a,
\qquad\mbox{where}\quad
D = \begin{pmatrix} 1 &0\\ 0 &\frac{1-\alpha+2\beta}{\alpha}\end{pmatrix},
\]
we can check that \eqref{e:vversion} holds with $\qt=D^{-1}q$, i.e., for all nonzero $v\in\R^N$,

\[
0<v^T(\Tilde{B}+a\one^T)v=v^T\begin{pmatrix} 2-\alpha & -2\beta\\ 2\beta &\beta\end{pmatrix}v.
\]
This implies $\Tilde{y}= \frac{\alpha}{1+\alpha+2\beta}(1,1,\frac{1-\alpha+2\beta}{\alpha})\in\Delta_2$, which corresponds to $\qt$, is an ESS. Note that $\Tilde{y}$ is dynamically equivalent to $y$.

Moreover, since an internal ESS is a global attractor, we can see that the Lotka-Volterra steady state state $q=\one$ is a global attractor in $\R^2_+$.

We can summarize the implications of this example as follows:
\begin{itemize}
    \item the converse of Theorem~\ref{t:ESSinLV} is false,
    \item the image of an ESS under a transformation of type (ii) might not be an ESS,
    \item Lemma~\ref{l:lvess} can be used to prove global stability of an internal equilibrium in a Lotka-Volterra system which is not VL-stable.
\end{itemize}
\end{exmp}

The next example shows that the conclusions that Theorem~\ref{t:ESSinLV} ensures for the Lotka-Volterra image of an ESS are sharp in the sense that we cannot expect strict stability of the corresponding equilibrium in general.

\begin{exmp} We describe an example with $N=3$ of a non-strictly stable steady state 
$q$ that corresponds to an ESS. 
Take
\begin{equation}
B = \begin{pmatrix} 1&1&1\\ 1& 1&2\\1&2&1 \end{pmatrix}, \qquad 
a = \begin{pmatrix} 1\\1\\1\end{pmatrix}, \qquad q = \begin{pmatrix}1\\0\\0\end{pmatrix}.
\end{equation}
Then $a=Bq$ and the condition in Lemma~\ref{l:lvess} reduces to saying that for all
nonzero $v=(v_1,v_2,v_3)$ with $v_2$, $v_3\ge0$,
\begin{equation}\label{e:expos}
0< v^T(B+\one\one^T)v = 2(v_1+v_2+v_3)^2 + 2v_2v_3  .
\end{equation}
This is indeed true, so $q$ does correspond to an ESS $y\in\Delta_3$ for the payoff matrix
in \eqref{e:Aform}.
The matrix $A$ coming from \eqref{e:linear}, the linearized Lotka-Volterra system about $\pt=q$,
takes the form
\[
A = \begin{pmatrix} -1 & -1 &-1\\ 0 &0&0\\0&0&0
\end{pmatrix}.
\]
This linearization is degenerate and $q$ is not strictly stable (linearly asymptotically stable).
We can note also that the matrix $B(I+q\one^T)$ is symmetric but \emph{not} positive definite,
despite the validity of \eqref{e:expos} when $v_2$, $v_3\ge0$.
\end{exmp}

{\it Relation to stability of cliques.}
One last comparison we will make is between 
our result on the stability of cliques in our graph-based Example~\ref{x:friends}
and the characterization of ESSs in terms of cliques 
by Cannings and Vickers \cite{vickers1988patterns} for payoff matrices 
with the same graph-based structure.

When the Lotka-Volterra growth rates $a_i$ are all the same, 
there is a different map between Lotka-Volterra solutions and the replicator equations~\cite[Exercise~7.5.2]{HofbauerSigmund}. 
Namely, this is the projection map $p\mapsto x\in\Delta_N$  given by 
\begin{align}\label{e:lv2rep2}
    x_i &= p_i/\sum_{j=1}^N p_j, \quad i=1,\cdots, N, 
\end{align}
together with a nonlinear time change, taking the payoff matrix $A$ simply as $-B$.

Theorem~1 of Cannings and Vickers states, in our present terminology, that
if $B=-A$ is as in Example~\ref{x:friends} above, so \eqref{c:abc}--\eqref{e:Brival} hold, then there is an ESS with support $T\subset\Nset$ if and only if $T$ is a clique.
Moreover, such an ESS must take the form $y=\one_T/|T|$. 
These ESSs comprise all the stable equilibria in the replicator equation in this case.

But as Proposition~\ref{p:clique} shows, a clique $T$ with $m=|T|$ members
supports a strictly stable state $p$ under Lotka-Volterra dynamics 
if and only if the additional condition $\beta+(m-1)\alpha>0$ from \eqref{d:cm} holds.
The case that is explicitly analyzed in \cite{vickers1988patterns} is 
$\alpha=-1$, $\beta=0$, in which case \eqref{d:cm} never holds 
and no strictly stable states exist.
(Any Lotka-Volterra solution with support inside a 
clique will be unbounded in time, in fact.)
Replicator dynamics remain invariant under 
adding the same constant to all entries of $A$, though. 
So after a suitable change of $\alpha$, $\beta$, $\gamma$ the ESSs and 
strictly stable Lotka-Volterra states can all correspond. 

\begin{remark}[The Cannings-Vickers characterization of ESSs] 
Here we address an issue in the proof of Theorem 1 in \cite{vickers1988patterns} and indicate a clarification. 

In the proof that the support $T$ of an ESS must be a clique,
Cannings and Vickers state that ``if $T$ is not a clique then there is a clique $T^*$ containing $T$, or contained in it." As a general statement about graphs, this is not true---E.g., the set $T=\{1,2,3\}$ in the following graph has no super- or sub-graph that is a clique:
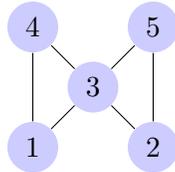
\begin{center}
\begin{tikzpicture}
  [scale=.8,auto=left,every node/.style={circle,fill=blue!20}]
  \node (n1) at (0,0)  {1};
  \node (n2) at (2,0)  {2};
  \node (n3) at (1,1)  {3};
  \node (n4) at (0,2)  {4};
  \node (n5) at (2,2)  {5};
 
  \foreach \from/\to in {n1/n4,n3/n5,n5/n2,n1/n3,n2/n3,n3/n4}
    \draw (\from) -- (\to);

\end{tikzpicture}
\end{center}
One can conclude $T$ is a clique by arguing as follows instead.
Suppose $T$ supports an ESS but not a clique. If $T$ is complete, we can find a clique $T^*$ that strictly contains $T$, 
which yields a contradiction with
the exclusion principle.
If $T$ is not complete, there exists a {complete} $T^*$ that is {maximal} as a \emph{subgraph} of $T$. Let $y$ be an ESS supported by $T$ and let $x=\frac{1}{|T^*|}\one_{T^*}$. 
Since $T^*\subset T$, $y^TAy=x^TAy$ so from Definition~\ref{d1:ess}(b), $y^TAx>x^TAx$ must hold. On the other hand, because $T^*$ is maximal in $T$,
\[
(Ax)_i
\begin{cases}
=\frac{1}{|T^*|}(|T^*|-1)& i\in T^*,\\
\leq\frac{1}{|T^*|}(|T^*|-2)& i\in T\setminus T^*.
\end{cases}
\]
This implies $x^TAx>y^TAx$, contradicting the observation we just made.
\end{remark}
    
\section*{Acknowledgments}
This material is based upon work supported by the National Science Foundation 
under grant DMS 1812609.

\bibliographystyle{siam}
\bibliography{selection}

\begin{thebibliography}{10}

\bibitem{armstrong1980competitive}
{\sc R.~Armstrong and R.~McGehee}, {\em Competitive exclusion}, Amer. Natur.,
  115 (1980), pp.~151--170.

\bibitem{bishop1976models}
{\sc D.~T. Bishop and C.~Cannings}, {\em Models of animal conflict}, Advances
  in Applied Probability, 8 (1976), pp.~616--621.

\bibitem{bomze2020ess}
{\sc I.~M. Bomze and W.~Schachinger}, {\em Constructing patterns of (many)
  {ESS}s under support size control}, Dyn. Games Appl., 10 (2020),
  pp.~618--640.

\bibitem{case1979global}
{\sc T.~J. Case and R.~G. Casten}, {\em Global stability and multiple domains
  of attraction in ecological systems}, Amer. Natur., 113 (1979), pp.~705--714.

\bibitem{chesson2000mechanisms}
{\sc P.~Chesson}, {\em Mechanisms of maintenance of species diversity}, Annual
  review of Ecology and Systematics, 31 (2000), pp.~343--366.

\bibitem{coyte2015ecology}
{\sc K.~Z. Coyte, J.~Schluter, and K.~R. Foster}, {\em The ecology of the
  microbiome: networks, competition, and stability}, Science, 350 (2015),
  pp.~663--666.

\bibitem{cross1978}
{\sc G.~Cross}, {\em Three types of matrix stability}, Linear Algebra and its
  Applications, 20 (1978), pp.~253 -- 263.

\bibitem{fiedler1966}
{\sc M.~Fiedler and V.~Pt\'ak}, {\em Some generalizations of positive
  definiteness and monotonicity.}, Numerische Mathematik, 9 (1966/67),
  pp.~163--172.

\bibitem{gause1934experimental}
{\sc G.~F. Gause et~al.}, {\em Experimental analysis of vito volterra’s
  mathematical theory of the struggle for existence}, Science, 79 (1934),
  pp.~16--17.

\bibitem{gilpin1975limit}
{\sc M.~E. Gilpin}, {\em Limit cycles in competition communities}, The American
  Naturalist, 109 (1975), pp.~51--60.

\bibitem{gilpin1976multiple}
{\sc M.~E. Gilpin and T.~J. Case}, {\em Multiple domains of attraction in
  competition communities}, Nature, 261 (1976), pp.~40--42.

\bibitem{goh1977global}
{\sc B.~S. Goh}, {\em Global stability in many-species systems}, The American
  Naturalist, 111 (1977), pp.~135--143.

\bibitem{goh1980management}
{\sc B.-S. Goh}, {\em Management and Analysis of Biological Populations},
  Elsevier, 1980.

\bibitem{goldford2018emergent}
{\sc J.~E. Goldford, N.~Lu, D.~Baji{\'c}, S.~Estrela, M.~Tikhonov,
  A.~Sanchez-Gorostiaga, D.~Segr{\`e}, P.~Mehta, and A.~Sanchez}, {\em Emergent
  simplicity in microbial community assembly}, Science, 361 (2018),
  pp.~469--474.

\bibitem{GolubVanLoan}
{\sc G.~H. Golub and C.~F. Van~Loan}, {\em Matrix Computations}, Johns Hopkins
  Studies in the Mathematical Sciences, Johns Hopkins University Press,
  Baltimore, MD, third~ed., 1996.

\bibitem{hadeler2017book}
{\sc K.~P. Hadeler}, {\em Topics in Mathematical Biology}, Lecture Notes on
  Mathematical Modelling in the Life Sciences, Springer, Cham, 2017.

\bibitem{hewer1988}
{\sc G.~Hewer and C.~Kenney}, {\em The sensitivity of the stable {L}yapunov
  equation}, SIAM J. Control Optim., 26 (1988), pp.~321--344.

\bibitem{HofbauerSigmund}
{\sc J.~Hofbauer and K.~Sigmund}, {\em Evolutionary Games and Population
  Dynamics}, Cambridge University Press, Cambridge, 1998.

\bibitem{hutchinson1957}
{\sc G.~E. Hutchinson}, {\em Population studies: Animal ecology and
  demography---concluding remarks}, Cold Spring Harbor Symposia on Quantitative
  Biology, 22 (1957), pp.~415--427.
\newblock Reprinted in: Bull. Math. Biol. 53 (1991) 193-213.

\bibitem{ings2009ecological}
{\sc T.~C. Ings, J.~M. Montoya, J.~Bascompte, N.~Bl{\"u}thgen, L.~Brown, C.~F.
  Dormann, F.~Edwards, D.~Figueroa, U.~Jacob, J.~I. Jones, et~al.}, {\em
  Ecological networks--beyond food webs}, Journal of Animal Ecology, 78 (2009),
  pp.~253--269.

\bibitem{kokkoris1999patterns}
{\sc G.~Kokkoris, A.~Troumbis, and J.~Lawton}, {\em Patterns of species
  interaction strength in assembled theoretical competition communities},
  Ecology Letters, 2 (1999), pp.~70--74.

\bibitem{kushel2019}
{\sc O.~Y. Kushel}, {\em Unifying matrix stability concepts with a view to
  applications}, SIAM Rev., 61 (2019), pp.~643--729.

\bibitem{lawmorton1993alternative}
{\sc R.~Law and R.~D. Morton}, {\em Alternative permanent states of ecological
  communities}, Ecology, 74 (1993), pp.~1347--1361.

\bibitem{lischke2017finding}
{\sc H.~Lischke and T.~J. L{\"o}ffler}, {\em Finding all multiple stable
  fixpoints of n-species {L}otka--{V}olterra competition models}, Theoretical
  Population Biology, 115 (2017), pp.~24--34.

\bibitem{LiuCaiSu2015}
{\sc H.~Liu, W.~Cai, and N.~Su}, {\em Entropy satisfying schemes for computing
  selection dynamics in competitive interactions}, SIAM J. Numer. Anal., 53
  (2015), pp.~1393--1417.

\bibitem{logofet2005}
{\sc D.~O. Logofet}, {\em Stronger-than-{L}yapunov notions of matrix stability,
  or how ``flowers'' help solve problems in mathematical ecology}, Linear
  Algebra Appl., 398 (2005), pp.~75--100.

\bibitem{lubell1966}
{\sc D.~Lubell}, {\em A short proof of {S}perner's lemma}, J. Combinatorial
  Theory, 1 (1966), p.~299.

\bibitem{may1977}
{\sc R.~M. May}, {\em Thresholds and breakpoints in ecosystems with a
  multiplicity of stable states}, Nature, 269 (1977), pp.~471--477.

\bibitem{may1975}
{\sc R.~M. May and W.~J. Leonard}, {\em Nonlinear aspects of competition
  between three species}, SIAM J. Appl. Math., 29 (1975), pp.~243--253.

\bibitem{mcgehee1977}
{\sc R.~McGehee and R.~A. Armstrong}, {\em Some mathematical problems
  concerning the ecological principle of competitive exclusion}, J.
  Differential Equations, 23 (1977), pp.~30--52.

\bibitem{moon1965}
{\sc J.~W. Moon and L.~Moser}, {\em On cliques in graphs}, Israel J. Math., 3
  (1965), pp.~23--28.

\bibitem{moylan1977}
{\sc P.~J. Moylan}, {\em Matrices with positive principal minors}, Linear
  Algebra Appl., 17 (1977), pp.~53--58.

\bibitem{pocheville2015}
{\sc A.~Pocheville}, {\em The ecological niche: History and recent
  controversies}, in Handbook of Evolutionary Thinking in the Sciences,
  T.~Heams, P.~Huneman, G.~Lecointre, and M.~Silberstein, eds., Springer
  Netherlands, Dordrecht, 2015, pp.~547--586.

\bibitem{smale1976}
{\sc S.~Smale}, {\em On the differential equations of species in competition},
  J. Math. Biol., 3 (1976), pp.~5--7.

\bibitem{smith1973logic}
{\sc J.~M. Smith and G.~R. Price}, {\em The logic of animal conflict}, Nature,
  246 (1973), pp.~15--18.

\bibitem{svirezhev2008nonlinearities}
{\sc Y.~M. Svirezhev}, {\em Nonlinearities in mathematical ecology: Phenomena
  and models: Would we live in {V}olterra's world?}, Ecological Modelling, 216
  (2008), pp.~89--101.

\bibitem{takeuchi}
{\sc Y.~Takeuchi}, {\em Global dynamical properties of {L}otka-{V}olterra
  systems}, World Scientific Publishing Co., Inc., River Edge, NJ, 1996.

\bibitem{tartar1971}
{\sc L.~Tartar}, {\em Une nouvelle caract\'{e}risation des {$M$} matrices},
  Rev. Fran\c{c}aise Informat. Recherche Op\'{e}rationelle, 5 (1971),
  pp.~127--128.

\bibitem{vickers1988patterns}
{\sc G.~T. Vickers and C.~Cannings}, {\em Patterns of {ESS}s. {I}, {II}}, J.
  Theoret. Biol., 132 (1988), pp.~387--408, 409--420.

\bibitem{volterra1928variations}
{\sc V.~Volterra}, {\em Variations and fluctuations of the number of
  individuals in animal species living together}, ICES Journal of Marine
  Science, 3 (1928), pp.~3--51.

\end{thebibliography}

\end{document}